\documentclass[12pt,onecolumn,peerreview,letterpaper]{IEEEtran}
\usepackage{amssymb}
\usepackage{amsmath}
\usepackage[dvips]{graphicx}
\usepackage{cite}
\usepackage{float}
\usepackage[T1]{fontenc}
\usepackage{array}
\usepackage[pdftex,bookmarks]{hyperref}

\newtheorem{lemma}{Lemma}

\newtheorem{theorem}{Theorem}
\newtheorem{corollary}{Corollary}
\newtheorem{example}{Example}

 \floatstyle{ruled}
\newfloat{algorithm}{tbp}{loa}
\floatname{algorithm}{Algorithm}

\begin{document}

\title{Lowering the Error Floor of LDPC Codes Using Cyclic Liftings}
\author{Reza ~Asvadi\authorrefmark{1},
        Amir ~H. Banihashemi\authorrefmark{2},~\IEEEmembership{Senior~Member,~IEEE,}
        \\ and Mahmoud Ahmadian-Attari\authorrefmark{1}

\authorblockA{\authorrefmark{1}Department of Electrical and Computer Engineering\\
K. N. ~Toosi University of Technology, Tehran, Iran\\
 Emails: asvadi@ee.kntu.ac.ir, mahmoud@eetd.kntu.ac.ir}\\
\authorblockN{\authorrefmark{2}Department of Systems and Computer Engineering, Carleton
University,\\ Ottawa, ON K1S 5B6, Canada\\
 Email: ahashemi@sce.carleton.ca}}


\maketitle

\renewcommand{\QED}{\QEDopen}

\begin{abstract}
Cyclic liftings are proposed to lower the error floor of low-density
parity-check (LDPC) codes. The liftings are designed to eliminate dominant
trapping sets of the base code by removing the short cycles
which form the trapping sets. We derive a necessary and sufficient condition
for the cyclic permutations assigned to the edges of a cycle $c$
of length $\ell(c)$ in the base graph such that the inverse image
of $c$ in the lifted graph consists of only cycles of length strictly
larger than $\ell(c)$. The proposed method is universal in the sense that it
can be applied to any LDPC code over any channel and for any iterative decoding algorithm.
It also preserves important properties of the base code such as
degree distributions, encoder and decoder structure, and in some cases,
the code rate.
The proposed method is applied to both structured and random codes
over the binary symmetric channel (BSC). The error floor improves consistently
by increasing the lifting degree, and the results show significant improvements
in the error floor compared to the base code, a random code of the
same degree distribution and block length, and a random lifting
of the same degree. Similar improvements are also observed when the
codes designed for the BSC are applied to the additive white
Gaussian noise (AWGN) channel.
\end{abstract}



\section{Introduction}


Low-density parity-check (LDPC) codes~\cite{Gal-63} have emerged as
one of the top contenders for capacity approaching error correction
over many important channels. They not only perform superbly but
also lend themselves well to highly efficient parallel decoding
algorithms. A well-known construction of LDPC codes is based on {\em
protographs}, also referred to as {\em base graphs} or {\em
projected graphs}~\cite{RU-Book}. In such constructions, a bipartite
base graph $G$ is copied $N$ times and for each edge $e$ of $G$, a
permutation is applied to the $N$ copies of $e$ to interconnect the
$N$ copies of $G$. The resulting graph, called the {\em $N$-cover}
or the {\em $N$-lifting} of $G$, is then used as the Tanner
graph~\cite{Tanner-81} of the LDPC code. If the permutations are
cyclic, the resulting LDPC code is called {\em quasi-cyclic} (QC).
QC LDPC codes are attractive due to their simple implementation and
analysis~\cite{RU-Book}.

At very large block lengths, the performance of LDPC codes can be well
estimated using asymptotic techniques such as density evolution~\cite{RU-Capacity}. At finite
lengths, however, our understanding of the dynamics of iterative decoding
algorithms is limited. In particular, iteratively decoded
finite-length codes demonstrate an abrupt change in their error rate curves,
referred to as {\em error floor}, in the high signal to noise ratio (SNR) region. The analysis of
the error floor and techniques to improve the error floor performance
of LDPC codes are still very active areas of research. For the binary
erasure channel (BEC), the error floor is well understood and is known
to be caused by graphical structures called {\em stopping sets}~\cite{Di-02}.
Richardson related the error rate performance of LDPC codes on
the binary symmetric channel (BSC) and the additive white Gaussian noise (AWGN)
channel to more general graphical structures, called {\em trapping sets},
and devised a technique to estimate the error floor~\cite{R-03}. Other estimation techniques
based on finding the dominant trapping sets were also proposed
for the BSC in~\cite{CSV-06} and for the AWGN channel in~\cite{Cole-06},~\cite{SC-06}.
In \cite{XB-07} - \cite{XB-09}, Xiao and Banihashemi took a different approach, and
instead of focusing on trapping sets
which are the eventual result of the decoder failure, focussed on
the input error patterns that cause the decoder to fail. A simple
technique for estimating the frame error rate (FER) and the bit error
rate (BER) of finite-length LDPC codes over the BSC was developed
in~\cite{XB-07}. The complexity of this algorithm was then reduced
in~\cite{XB-09}, and the estimation technique was extended to the AWGN channel with
quantized output in~\cite{XB-08}. More recent work on the estimation of
the error floor of LDPC codes is presented in~\cite{CCSV-09},~\cite{DLZANW-09},
to which the reader is
also referred for a more comprehensive list of references.

There is extensive literature on reducing the error floor of finite-length
LDPC codes over different channels and for different iterative decoding
algorithms. One category of such literature, focusses on modification of
iterative decoding algorithms, see, e.g.,~\cite{HR-09}, while another category
is concerned with the code construction. In the second category, some
researchers use indirect measures such as girth~\cite{TSF-01} or approximate
cycle extrinsic message degree (ACE)~\cite{VS-09}, while
others work with direct measures of error floor performance
such as the distribution of stopping sets or trapping sets~\cite{Wang-06},~\cite{ICV-08},~\cite{JMSZ-09}.
In~\cite{Wang-06}, {\em edge swapping} is proposed as a technique to increase the
stopping distance of an LDPC code, and thus to
improve its error floor performance over the BEC. Random
cyclic liftings are also studied in~\cite{Wang-06} and shown to improve
the average performance of the ensemble in the error floor
region compared to the base code.
Ivkovic {\em et al.}~\cite{ICV-08} apply the same technique of edge swapping
between two copies of a base LDPC code to eliminate
the dominant trapping sets of the base code over the BSC.

In the approach proposed here also, we focus on dominant trapping sets which
are the main contributors to the error floor.
We start from the code whose error floor is to be improved,
as the base code. We then construct a new code by cyclically lifting
the base code. The lifting is designed carefully to eliminate the
dominant trapping sets of the base graph. This is achieved by removing
the short cycles which form the dominant trapping sets.
Our work has similarities to~\cite{ICV-08} and~\cite{Wang-06}.
The similarity with both~\cite{ICV-08} and~\cite{Wang-06} is that we also use graph covers or liftings to
improve the error floor performance of a base code. It however differs from~\cite{ICV-08} in
that we restrict ourselves to cyclic liftings that are advantageous
in implementation. Moreover, to eliminate the dominant trapping sets
we use a different approach than the one in~\cite{ICV-08}. More specifically,
our approach is based on the elimination of the short cycles involved
in the trapping sets. To do so, we derive a necessary and sufficient condition
for the problematic cycles of the base code such that they are mapped to strictly
larger cycles in the lifted code. The difference with~\cite{Wang-06} is that
while~\cite{Wang-06} is focused on
the {\em ensemble} performance of {\em random} liftings, our work is concerned
with the {\em intentional} design of a {\em particular} cyclic lifting.

Given a base code and its dominant trapping sets over a certain channel
and under a specific iterative decoding algorithm, the proposed construction
can lower the error floor by increasing the block length while preserving
the important properties of the base code such as degree distributions, and the
encoder and decoder structure. The code rate is also
preserved or is decreased slightly depending on the rank deficiency of
the parity-check matrix of the base code. Moreover, the cyclic
nature of the lifting makes it implementation friendly. We apply the proposed
construction to the Tanner code~\cite{TSF-01} and two randomly constructed codes,
one regular and the other irregular, to improve
the error floor performance of Gallager A/B algorithms over the BSC.\footnote{The
choice of BSC/Gallager algorithms is for simplicity,
and the proposed construction is applicable to any channel/decoding algorithm
combination as long as the dominant trapping sets are known.}
Simulation results show a consistent improvement in the error floor performance by
increasing the degree of liftings. The constructed codes are far superior to
similar random codes or codes constructed by random liftings in the error floor region.
We also examine the performance of the codes constructed for
BSC/Gallager B algorithm, over the AWGN channel with min-sum decoding
and observe similar improvements in the error floor performance.

The remaining of the paper is organized as follows: Section
\ref{sec::prelimin} introduces definitions, notations and
background material used throughout the paper. In Section III,
the proposed construction is explained and discussed. Numerical results are
presented in Section \ref{sec::simul}, and finally Section
\ref{sec::conclusion} concludes the paper.

\section{PRELIMINARIES: LDPC CODES, TANNER GRAPHS, GRAPH LIFTINGS
AND TRAPPING SETS}
\label{sec::prelimin}
\subsection{LDPC Codes and Tanner Graphs}
Consider a binary LDPC code ${\cal C}$ represented by a Tanner graph
$G = (V_b \cup V_c , E)$, where $V_b = \{b_1, \ldots, b_n\}$ and $V_c = \{c_1, \ldots, c_m\}$
are the sets of variable nodes and check nodes, respectively, and $E$ is the set of edges.
Corresponding to $G$, we have an $m \times n$ parity-check matrix
$H=[h_{ij}]$ of ${\cal C}$, where $h_{ij} = 1$ if and only if (iff) the
node $c_i \in V_c$ is connected to the node $b_j \in V_b$ in $G$;
or equivalently, iff $\{b_j,c_i\} \in E$. If all the nodes in the
set $V_b$ have the same degree $d_v$ and all the nodes
in the set $V_c$ have the same degree $d_c$, the corresponding
LDPC code is called a {\em regular} $(d_v,d_c)$ code. Otherwise,
it is called {\em irregular}.

A subgraph of $G$ is a {\em path} of length $k$ if
it consists of a sequence of $k+1$ nodes $\{u_1,\ldots,u_{k+1}\}$
and $k$ distinct edges $\{ \{u_i,u_{i+1}\} : i=1,\ldots,k \}$. We say two
nodes are {\em connected} if there is a path between them.
A path is a {\em cycle} if $u_1 = u_{k+1}$, and all the other nodes are distinct. The length of the
shortest cycle(s) in the graph is called {\em girth}. In bipartite graphs, including
Tanner graphs, all cycles have even lengths. So, the girth is an even number.

\subsection{Graph Liftings}
\label{subsec2.2}
Consider the set of all possible permutations $S_N$ over the set of
integer numbers $Z_{1 \rightarrow N} \stackrel{\Delta}{=} \{1, \ldots , N\}$.
This set forms a group, known as the {\em symmetric group}, under composition.
Each element $\pi \in S_N$
can be represented by all the values $\pi(i),\: i \in Z_{1 \rightarrow N}$.
For the identity element $\pi_0$, we have $\pi_0(i)=i, \forall i$, and the
inverse of $\pi$ is denoted by $\pi^{-1}$ and defined as $\pi^{-1}(\pi)=\pi_0$.
It is easy to see that the symmetric group is not Abelian.
An alternate representation of permutations is to represent a permutation
$\pi$ with an $N \times N$ matrix $\Pi = [\pi_{ij}]$, whose elements are
defined by $\pi_{ij} = 1$ if $j = \pi(i)$, and $\pi_{ij} = 0$, otherwise.
As a result, we have the isomorphic group of all $N \times N$
permutation matrices with the group operation defined as matrix
multiplication, and the identity element equal to the identity
matrix $I_N$. In particular, corresponding to the composition
$\pi'(\pi)$, we have the matrix multiplication $\Pi \times \Pi'$.
Moreover, since the permutation matrices are orthogonal,
the inverse of a permutation matrix is its transpose,
i.e., $\Pi^{-1} = \Pi^T$.

Consider the cyclic subgroup $C_N$ of $S_N$ consisting of the $N$ circulant
permutations defined by $\pi_d(i) = i+d \:\:\:\;\mbox{mod}\:\:\:\; N + 1, \: d \in Z_{0 \rightarrow N-1}$.
The permutation $\pi_d$ corresponds to $d$ cyclic shifts to the right.
In the matrix representation, permutation $\pi_d$ corresponds to a permutation matrix
whose rows are obtained by cyclically shifting all the rows of the identity matrix $I_N$ by $d$
to the right. This matrix is denoted by $I^{(d)}$.
Note that $I^{(0)} = I_N$. For the composition of permutations,
we have
\begin{equation}
I^{(d_1)} \times I^{(d_2)} = I^{(d_1 + d_2 \:\mbox{mod}\; N)}\:.
\label{eqa}
\end{equation}

Clearly, $C_N$ can be generated by $I^{(1)}$, where each element $I^{(d)}, \: d \in Z_{0 \rightarrow N-1}$
of $C_N$ is the $d$-th power of $I^{(1)}$. This defines a natural isomorphism
between $C_N$ and the set of integers modulo $N$, $Z_{0 \rightarrow N-1}$, under addition.
The latter group is denoted by $Z_N$.

Consider the following construction of a graph $\tilde{G} = (\tilde{V},\tilde{E})$
from a graph $G = (V,E)$: We first make $N$ copies of $G$ such that for each node $v \in V$,
we have $N$ copies $\tilde{v} \stackrel{\Delta}{=} \{v_1 \ldots, v_N\}$ in $\tilde{V}$.
For each edge $e = \{u , v\} \in E$, we assign
a permutation $\pi^e \in S_N$ to the $N$ copies of $e$ in $\tilde{E}$ such that
an edge $\{u_i , v_j\}$ belongs to $\tilde{E}$ iff $\pi^e(i)=j$. The set
of these edges is denoted by $\tilde{e}$. The graph
$\tilde{G}$ is called an $N$-{\em cover} or an $N$-{\em lifting} of $G$, and $G$ is referred to
as the {\em base graph}, {\em protograph} or {\em projected graph} corresponding to $\tilde{G}$.
We also call the application of a permutation $\pi^e$ to the $N$ copies of $e$,
{\em edge swapping}, high lighting the fact that the permutation
swaps edges among the $N$ copies of the base graph.

In this work, $G$ is a Tanner graph, and we define the edge permutations from
the variable side to the check side, i.e., the set of edges $\tilde{e}$ in $\tilde{E}$ corresponding to
an edge $e = \{b,c\} \in E$ are defined by $\{b_i,c_{\pi^e(i)}\}, \: i \in Z_{1 \rightarrow N}$.
Equivalently, $\tilde{e}$ can be described by $\{b_{(\pi^e)^{-1}(j)},c_j\}, \: j \in Z_{1 \rightarrow N}$.
Our focus in this paper is on {\em cyclic liftings} of $G$, where the edge permutations
are selected from $C_N$, or equivalently $Z_N$. Thus the nomenclature {\em cyclic edge swapping}.
In this case, if $I^{(d)}$ is a permutation matrix from variable nodes to check nodes,
$I^{(d')};\:d'= N-d\:\:\mbox{mod}\:\:N$, will be the corresponding permutation matrix from check nodes
to variable nodes. It is important to distinguish between the two cases when we
compose permutations on a directed path.

To the lifted graph $\tilde{G}$, we associate an
LDPC code $\tilde{\cal{C}}$, referred to as the {\em lifted code},
such that the $mN \times nN$ parity-check matrix $\tilde{H}$ of $\tilde{\cal{C}}$
is equal to the adjacency matrix of $\tilde{G}$. More specifically, $\tilde{H}$
consists of $m \times n$ sub-matrices $[\tilde{H}]_{ij}\:, 1 \leq i \leq m, \: 1 \leq j \leq n$,
arranged in $m$ rows and $n$ columns. The
sub-matrix $[\tilde{H}]_{ij}$ in row $i$ and column $j$ is the permutation matrix from $C_N$
corresponding to the edge $\{b_j,c_i\}$ where $h_{ij} \neq 0$; otherwise, $[\tilde{H}]_{ij}$ is the all-zero matrix.
Let the $m \times n$ matrix $D=[d_{ij}]$ be defined
by $[\tilde{H}]_{ij} = I^{(d_{ij})},\:d_{ij} \in Z_N$ if $h_{ij} \neq 0$, and $d_{ij} = +\infty$, otherwise.
Matrix $D$, called the matrix of {\em edge permutation indices}, fully describes $\tilde{H}$ and thus the
cyclically lifted code $\tilde{\cal{C}}$.

\subsection{Trapping Sets and Error Floor}
It is well-known that trapping sets are the culprits in the error floor region of
iterative decoding algorithms. An $(a,b)$ trapping set is defined as a set of $a$
variable nodes which have $b$ check nodes of odd degree in their induced subgraph.
Among trapping sets, the most harmful ones are called {\em dominant}.
Trapping sets depend not only on the Tanner graph of the code but also
on the channel and the iterative decoding algorithm. In general, finding
all the trapping sets is a hard problem, and one often needs to resort to
efficient search techniques to obtain the dominant trapping
sets~\cite{R-03},~\cite{Cole-06},~\cite{XB-09},~\cite{WKP-09}.
Trapping sets for Gallager A/B algorithms over the BSC are examined for a number of
LDPC codes in~\cite{CSV-06},~\cite{XB-09}. In this work, we assume that the
dominant trapping sets are known and
available.

In the context of symmetric decoders over the BSC, the error floor of FER
can be estimated by~\cite{XB-07}
\begin{equation}
FER \approx N_J \:\epsilon^J\:,
\label{eqsd}
\end{equation}
where $\epsilon$ is the channel crossover probability, and
$J$ and $N_J$ are the size and the number of the smallest error patterns
that the decoder fails to correct. From (\ref{eqsd}), it is clear that the
dominant trapping sets over the BSC are those caused by the minimum
number of initial errors. (In~\cite{CSV-06},~\cite{ICV-08}, parameter $J$ in (\ref{eqsd})
is called {\em minimum critical number}.) In the double-logarithmic plane,
one can see from (\ref{eqsd}) that $\log(FER)$ decreases linearly with
$\log(\epsilon)$ and the slope of the line is determined by $J$.

\section{DESIGN OF CYCLIC LIFTINGS TO ELIMINATE TRAPPING SETS}

In this work, our focus is on the design of cyclic liftings of a given Tanner graph
to eliminate its dominant trapping sets with respect to a given channel/decoding algorithm
with the purpose of reducing the error floor.
(This is equivalent to the design of non-infinity edge permutation indices of matrix $D$.)
For example, Equation (\ref{eqsd})
indicates that for improving the error floor on the BSC, one needs
to increase $J$ and decrease $N_J$ corresponding to the dominant trapping sets.
In particular, while increasing $J$ for dominant trapping sets increases the slope of
$\log(\mbox{FER})$ vs. $\log(\epsilon)$ at low channel crossover
probabilities $\epsilon$, reducing $N_J$ amounts to a downward shift of the curve.
So, the general idea is to primarily increase $J$ by eliminating the trapping
sets with the smallest critical number, and the secondary goal is then
to decrease $N_J$.

It is well-known that dominant trapping sets composed of short cycles~\cite{CSV-06},~\cite{XB-09}.
To eliminate the trapping sets, we thus aim at eliminating their constituent cycles
in the lifted graph. In the following, we examine the inverse image of a (base)
cycle in the cyclically lifted graph.

\subsection{Cyclic Liftings of Cycles}
\label{subsec3.1}
\begin{lemma}
Consider a cyclic $N$-lifting $\tilde{G}$ of a Tanner graph $G$.
Consider a path $\xi$ of length $\ell$ in $G$, which starts from a variable node
$b$ and ends at a variable node $b'$ with the sequence of edges
$e_1, \ldots, e_{\ell}$. Corresponding to the edges, we have the sequence
of permutation matrices $I^{(d_1)},\ldots,I^{(d_\ell)}$. Then
the permutation matrix that maps $\tilde{b}$ to $\tilde{b'}$
through the path $\tilde{\xi}$ is $I^{(d)}$, where
\begin{equation}
d=\sum_{i=0}^{\ell-1} (-1)^i d_{i+1} \:\mbox{mod}\: N\:.
\label{eq2}
\end{equation}
\label{lem1}
\end{lemma}
\begin{proof} The permutation matrix that maps $\tilde{b}$ to $\tilde{b'}$
is obtained by multiplying the permutation matrices of the ``directed" edges
along the path. This results in $I^{(d)} = I^{(d_1)} \times I^{(N - d_2 \:\mbox{mod}\:N)} \times
\cdots \times I^{(N - d_\ell\:\mbox{mod}\:N)}$. The lemma is then proved using (\ref{eqa}).
\end{proof}

The value of $d$ given in (\ref{eq2}) is called the {\em permutation index} of the
path from $b$ to $b'$. Clearly, the permutation index of the path from $b'$ to $b$ is equal to
$d' = N-d\:\:\mbox{mod}\:\:N$. If $b = b'$ and all the other nodes are distinct, then the path will become a cycle
and depending on the direction of the cycle,
its permutation index will be equal to $d$ or $d'$.

\begin{theorem} \label{thm1}
Consider the cyclic $N$-lifting $\tilde{G}$ of the Tanner graph $G$.
Suppose that $c$ is a cycle of length $\ell$ in $G$. The inverse image of
$c$ in $\tilde{G}$ is then the union of $N/k$ cycles, each of length $k\ell$,
where $k$ is the order of the element(s) of $Z_N$ corresponding to the
permutation indices of $c$.
\end{theorem}
\begin{proof}
We first note that the elements of $Z_N$ corresponding to both permutation indices of $c$
have the same order. Suppose that the permutation indices of $c$ in the two directions are
equal to $d$ and $d' = N-d\:\:\mbox{mod}\:\:N$. Now, $d'$ is the inverse of $d$ in $Z_N$,
and thus has the same order as $d$.

Denote the sequence of variable and check nodes in $c$ by $b_{i_1}, c_{i_2}, b_{i_3}, c_{i_4},
\ldots, b_{i_{\ell - 1}}, c_{i_{\ell}}, b_{i_1}$.
Starting from any of the $N$ variable nodes
in $\tilde{b_{i_1}}$, say $b_{i_{1}j}, \: 1 \leq j \leq N$, without loss of generality,
we follow one of the two paths of length $\ell$ in the inverse
image of $c$ in $\tilde{G}$, which corresponds to the direction on $c$ associated with the permutation index $d$.
As a result, based on Lemma~\ref{lem1}, we will end up at the variable node $b_{i_{1}p}$ in $\tilde{b_{i_1}}$, where
$p = j + d \:\:\:\mbox{mod}\:\:\:N$ and $d$ is given in (\ref{eq2}).
If $d = 0$, then the path ends at $b_{i_{1}j}$,
meaning that the inverse image of $c$ which passes through $b_{i_{1}j}$ is a cycle of length $\ell$.
Similarly one can see that the inverse image of $c$ passing through any of the $N$
nodes $b_{i_{1}1},\ldots,b_{i_1N}$ is a cycle of length $\ell$, and since these cycles do not overlap,
the inverse image of $c$ in this case is the union of $N$ cycles, each of length $\ell$.
This corresponds to the case where the element of $Z_N$ corresponding to the
permutation index of $c$ is $0$, which has order $k = 1$.
For the cases where $d \neq 0$, continuing along the path and passing through
$q \ell$ edges of the inverse image of $c$, we reach the node $b_{i_{1}p'}$,
where $p' = j + qd \:\:\:\mbox{mod}\:\:\:N$. Clearly, we will be back to the starting node for
the first time when $qd = 0 \:\:\mbox{mod}\:\:N$, for the smallest value of $q$.
(In that case by continuing the path we will
just go over the same cycle of length $q \ell$.) By definition, $q$ is the order of $d$ in $Z_N$,
and thus $q = k$. Since the order of any element of a finite group divides
the order of the group, $s = N/k$ is an integer. It can then be easily seen that
by starting from any of the nodes $b_{i_{1}1},\ldots,b_{i_1s}$,
we can partition the inverse image of $c$ into $s$ cycles of length $k \ell$ each. This completes the proof.
\end{proof}

In what follows, we refer to the value $k$ in Theorem~\ref{thm1}
as the {\em order} of cycle $c$, and use the notation ${\cal O}(c)$ to denote it.

\begin{corollary}
Consider the cyclic $N$-lifting $\tilde{G}$ of the Tanner graph $G$.
Suppose that $c$ is a cycle of length $\ell$ in $G$. The inverse image of
$c$ in $\tilde{G}$ is the union of non-overlapping cycles,
each strictly longer than $\ell$ iff ${\cal O}(c) > 1$; or equivalently,
iff the permutation index of $c$, given in~(\ref{eq2}), is nonzero.
\label{cor1}
\end{corollary}

\subsection{Intentional Edge Swapping (IES) Algorithm}
\label{subsec3.2}

Suppose that $T$ is the set of all dominant trapping sets, and $C(T)$ is the set of
all the cycles in $T$. We also use the notations $t$ and $C(t)$ for a trapping set and its constituent cycles,
respectively. For an edge $e$, we use $C^e(t)$ to denote the set of cycles in
the trapping set $t$ that include $e$. In the previous subsection, we proved that a cycle $c$ in the base
Tanner graph $G$ is mapped to the union of larger cycles in the cyclically lifted graph
$\tilde{G}$ iff ${\cal O}(c) > 1$. To eliminate the dominant trapping sets,
we are thus interested in assigning the edge permutation indices to the edges of $C(T)$
such that ${\cal O}(c) > 1$ for every cycle $c \in C(T)$. To achieve this, we order the trapping
sets in accordance with the increasing order of their critical number. We still denote this ordered set by $T$
with a slight abuse of notation. Note that $T$ may now include trapping sets with critical numbers
larger than the minimum one. We then go through the trapping sets in $T$ one at a time and
identify and list all the cycles involved in each trapping set in $C(T)$, i.e., $C(T)=\{c \in C(t), \forall t \in T\}$.
Note that $C(T)$ is also partially ordered based on the corresponding ordering of the trapping sets in $T$.
\begin{example}
\label{ex1}
Three typical trapping sets for Gallager A/B algorithms are shown in Fig.~1~\cite{CSV-06}.
The $(4,4)$ and $(5,3)$ trapping sets include one and three cycles of length 8, respectively,
while the $(4,2)$ trapping set has 2 cycles of length 6 and one cycle of length 8.
\end{example}

\begin{figure}[h]
\centering
\includegraphics[width = 0.54\textwidth]{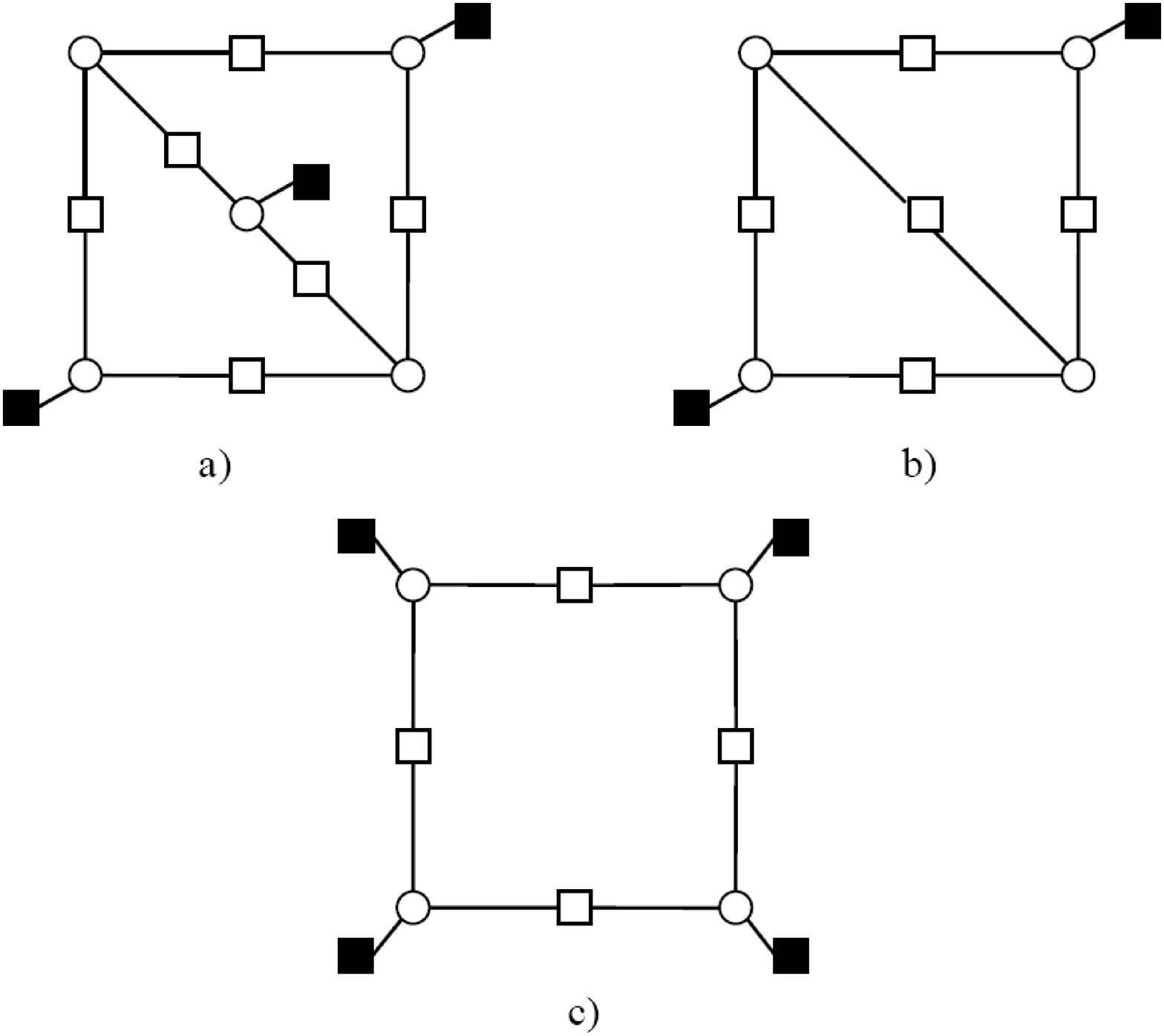}
\caption{a) (5, 3) Trapping set b) (4, 2) Trapping set c) (4, 4)
Trapping set. $\circ = $ Variable Node, $\square =$ Even degree
Check Node and $\blacksquare =$ Odd degree Check Node.}
\label{fig::trappingset}
\end{figure}

The next step is to choose proper edges of each trapping set to be swapped,
i.e., to choose the edges to which nonzero permutation indices are assigned.
In general, the policy is to select the minimum number of edges that can
result in ${\cal O}(c) > 1$ for every cycle $c$ in the trapping set $t$ under
consideration.
\begin{example}
\label{ex2}
Going back to Fig.~1, for the $(4,4)$ trapping set, it would be enough to just pick
one of the edges of the cycle of length 8 to eliminate this cycle in the lifted graph.
For the $(5,3)$ and $(4,2)$ trapping sets, however, at least two edges should be
selected for the elimination of all the cycles. A proper choice would be to select
one edge from the diagonal and the other edge from one of the sides.
\end{example}
Related to the edge selection, is the next step of permutation index assignment
to the selected edges such that ${\cal O}(c) > 1$ for every cycle $c \in C(t)$.
In general, we would like to have larger orders for the cycles. This in turn would result in
larger cycles in the lifted graph. To limit the complexity, however, we approach this
problem in a greedy fashion and with the main goal of just eliminating
all the cycles in $C(t)$, i.e., for each selected edge $e$, we choose the permutation index
such that all the cycles $C^e(t)$ have orders larger than one. This can
be performed by sequentially testing the values in the set $Z_{1 \rightarrow N}$.\footnote{More
complex search algorithms with the goal of increasing the order of cycles may be devised.
In this work however, no attempt has been made in this direction.}
As soon as such an index is found, we assign it to $e$ and move to the next
selected edge and repeat the same process.

We call the proposed algorithm {\em intentional edge swapping (IES)}
to distinguish it from ``random edge swapping," commonly used
to construct lifted codes and graphs. The pseudocode of the algorithm
is given as Algorithm 1. At the output of Algorithm 1, we have the sets
$SwappedSet$ and $IndexSet$, which contain the edges of the Tanner graph that
should be swapped, and their corresponding permutation indices,
respectively.

\begin{algorithm}[h]
1) \textbf{Initialization:} Create the ordered sets $T$ and $C(T)$. Select $N$.
$ProcessedSet = \emptyset$, $SwappedSet = \emptyset$, $IndexSet=\emptyset$.\\
2) Select the next trapping set $t \in T$.\\
3) $CurrentSet =$ edges of $C(t)$.\\
4) $CandidateSet = CurrentSet  \setminus ProcessedSet$. \\
5) If $CandidateSet = \emptyset$, go to Step 8.\\
6) Select the edges ${\cal E}$ from $CandidateSet$ that should be swapped, and assign
their permutation indices ${\cal I}$ from $Z_{1 \rightarrow N}$
such that ${\cal O}(c) > 1$ for every cycle $c$ in $C^e(t)$. \\
7) $SwappedSet = SwappedSet \:\cup\: {\cal E}$, $IndexSet = IndexSet \:\cup\: {\cal I}$,
and $ProcessedSet = ProcessedSet \:\cup\: CurrentSet$. Go to Step 12.\\
8) $CandidateSet = CurrentSet \setminus SwappedSet$. If $CandidateSet = \emptyset$, Stop.\\
9) Select an edge $e$ from $CandidateSet$ and assign a permutation index $i \in Z_{1 \rightarrow N}$ to it such that
for all cycles $c \in C^e(ProcessedSet \:\cup\: t)$, we have ${\cal O}(c) > 1$. If this is not feasible,
go to Step 11.\\
10) $SwappedSet = SwappedSet \cup e$, $IndexSet = IndexSet \cup i$,
and $CurrentSet = CurrentSet \setminus C^e(t)$. If $CurrentSet \neq \emptyset$,
go to Step 8. Otherwise, go to Step 12.\\
11) $CandidateSet = CandidateSet \setminus e$, If $CandidateSet \neq \emptyset$, go to Step 9.
Else, stop.\\
12) If all the trapping sets in $T$ are processed, stop. Otherwise, go to Step 2.
\caption{Intentional Edge Swapping (IES) Algorithm}
\label{alg1}
\end{algorithm}

In Algorithm~\ref{alg1}, the search for edges to be swapped and the
permutation index assignment to these edges are performed in two phases.
The first phase is in Steps 4 - 6, where any edge from
previously processed trapping sets is removed from the set of candidates for swapping.
If the first phase fails, in that no edge exists as a candidate for swapping
($CandidateSet = \emptyset$), then the algorithm switches to
the second phase in Steps 8 - 9, where only previously swapped edges
are removed from the candidate set for swapping.

The process of permutation index assignment in Algorithm~\ref{alg1} involves
the satisfaction of inequalities $d \neq 0$ for certain cycles,
where $d$ is given in (\ref{eq2}). In general, this is easier to achieve
if the variables involved in (\ref{eq2}) are selected from a larger alphabet space.
In fact, by increasing $N$, one can eliminate more trapping sets and achieve a
better performance in the error floor region.
\subsection{Minimum Distance and Rate of Cyclic Liftings}
\label{subsec3.3}
Consider an LDPC code $\cal{C}$ with an $m \times n$ parity-check matrix $H$.
To prove our results on the minimum distance and the rate of a cyclic $N$-lifting
$\tilde{\cal{C}}$ of $\cal{C}$, we consider an alternate parity-check matrix
$\tilde{H}'$ of $\tilde{\cal{C}}$ obtained by permutations of rows and columns of
matrix $\tilde{H}$ introduced in Subsection~\ref{subsec2.2}, as follows:
\begin{equation}
\label{eq3}
{\tilde{H}'}=
\left (\begin{array}{cccccc}
\mathcal{A}_{0} & \mathcal{A}_{N-1} &
\mathcal{A}_{N-2} & \ldots & \mathcal{A}_{2} & \mathcal{A}_{1}\\
\mathcal{A}_{1} & \mathcal{A}_{0} & \mathcal{A}_{N-1} &
\ldots & \mathcal{A}_{3}& \mathcal{A}_{2}\\
\vdots & \vdots & \vdots & \ddots & \vdots & \vdots\\
 \mathcal{A}_{N-2} & \mathcal{A}_{N-3} & \mathcal{A}_{N-4} & \ldots & \mathcal{A}_{0} & \mathcal{A}_{N-1}\\
\mathcal{A}_{N-1} & \mathcal{A}_{N-2} & \mathcal{A}_{N-3} & \ldots &
\mathcal{A}_{1} & \mathcal{A}_{0}
\end{array}
  \right)\:.
\end{equation}

In (\ref{eq3}), all the sub matrices $\mathcal{A}_{d},\: d = 0, \ldots, N-1$, have size
$m \times n$, and are given by
\begin{equation}
\label{eq5}
(i,j)\text{th entry of } \mathcal{A}_d= \left \{
\begin{array}{cc}
1 & \textrm{if $d = d_{ij}$} \\
0 & \quad \textrm{otherwise}, \\
\end{array}\right.
\end{equation}
where $d_{ij}$ is the permutation index corresponding to $h_{ij}$.
The parity-check matrix $\tilde{H}'$ is block circulant with the property that
\begin{equation}
\mathcal{A}_{0}+\mathcal{A}_{1}+\ldots+\mathcal{A}_{N-1} = H.
\label{eq8}
\end{equation}

\begin{theorem}
If code $\cal{C}$ has rate $r$, then the rate
$r^{(N)}$ of a cyclic $N$-lifting $\tilde{\cal{C}}$ of $\cal{C}$ satisfies $r^{(N)} \leq r$,
for $N = 2^q,\:q=1, 2, 3, \ldots$.
\label{thmz}
\end{theorem}

\begin{proof}
Due to the block circulant structure of $\tilde{H}'$, it can be written as
\begin{equation}
\tilde{H}'= \left
(\begin{array}{c c} \mathcal{M} & \mathcal{N} \\ \mathcal{N} &
\mathcal{M}
\end{array} \right)\:,
\label{eqdr}
\end{equation}
where matrices $\mathcal{M}$ and $\mathcal{N}$ are given by
$$
\mathcal{M} = \left (\begin{array}{ccccc}
\mathcal{A}_{0} & \mathcal{A}_{N-1} &
 \ldots & \mathcal{A}_{\frac{N}{2}+2} & \mathcal{A}_{\frac{N}{2}+1}\\ \mathcal{A}_{1} & \mathcal{A}_{0} &
 \ldots & \mathcal{A}_{\frac{N}{2}+3}& \mathcal{A}_{\frac{N}{2}+2} \\ \vdots & \vdots & \ldots & \vdots & \vdots\\
\mathcal{A}_{\frac{N}{2}+1} &  \mathcal{A}_{\frac{N}{2}}  &
 \ldots & \mathcal{A}_{1}& \mathcal{A}_{0} \end{array}
  \right)_{\frac{N}{2}m\times \frac{N}{2}n},
$$
and
$$
\mathcal{N} = \left (\begin{array}{ccccc}
\mathcal{A}_{\frac{N}{2}} & \mathcal{A}_{\frac{N}{2}-1} &
 \ldots & \mathcal{A}_{2}&\mathcal{A}_{1}\\ \mathcal{A}_{\frac{N}{2}+1} & \mathcal{A}_{\frac{N}{2}} &
 \ldots &\mathcal{A}_{3}& \mathcal{A}_{2} \\ \vdots & \vdots & \ldots & \vdots & \vdots\\ \mathcal{A}_{N-1} &  \mathcal{A}_{N-2}  &
 \ldots & \mathcal{A}_{\frac{N}{2}+3}& \mathcal{A}_{\frac{N}{2}+2} \end{array}
  \right)_{\frac{N}{2}m\times \frac{N}{2}n},
$$
respectively. (Note that all indices $i$ of $\mathcal{A}_i$ should be interpreted
as modulo $N$.) Adding the second block column of (\ref{eqdr}) to the first, followed by adding
the first block row to the second, we have
\begin{equation}
 \left(\begin{array}{c c} \mathcal{M} & \mathcal{N} \\ \mathcal{N} &
\mathcal{M} \end{array} \right) \rightarrow  \left (\begin{array}{c
c} \mathcal{M}+ \mathcal{N} & \mathcal{N} \\ \mathcal{N}+\mathcal{M} &
\mathcal{M} \end{array} \right) \rightarrow  \left (\begin{array}{c
c} \mathcal{M}+\mathcal{N} & \mathcal{N} \\ 0 &
\mathcal{M}+\mathcal{N}
\end{array} \right)\:.
\label{eqac}
\end{equation}
For $N=2$, $\mathcal{M} + \mathcal{N} = H$, and since the rank of the matrix in (\ref{eqac}),
and thus the rank of $\tilde{H}'$,
is at least twice the rank of $H$, we have $r^{(2)} \leq r$, and the proof is complete.
For $N > 2$, it is easy to see that $\mathcal{M} + \mathcal{N}$ is also block circulant and can in turn
be partitioned into four sub matrices, each of size $\frac{N}{4}m \times \frac{N}{4}n$, as follows:
$$
\mathcal{M}+\mathcal{N}=\left
(\begin{array}{c c} \mathcal{M}' & \mathcal{N}' \\ \mathcal{N}' &
\mathcal{M}'\end{array} \right)\:.
$$
Replacing this in the rightmost matrix of (\ref{eqac}), and performing similar block operations
as in (\ref{eqac}), we obtain
$$
\left(\begin{array}{c | c} \begin{array}{c c} \mathcal{M}' & \mathcal{N}'\\ \mathcal{N}' & \mathcal{M}'\end{array} & \mathcal{N}
\\ \hline 0 & \begin{array}{c c} \mathcal{M}' & \mathcal{N}'\\
\mathcal{N}' & \mathcal{M}'\end{array}
\end{array} \right) \rightarrow
$$
\begin{equation}
\left(\begin{array}{c | c} \begin{array}{c c} \mathcal{M}'+\mathcal{N}' & \mathcal{N}'\\
0 & \mathcal{M}'+\mathcal{N}'\end{array} & \mathcal{P} \\
\hline 0 & \begin{array}{c c} \mathcal{M}'+\mathcal{N}' &
\mathcal{N}'\\ 0 & \mathcal{M}'+\mathcal{N}'\end{array}
\end{array} \right)\:.
\label{eqjh}
\end{equation}
For $N = 4$, $\mathcal{M}'+\mathcal{N}' = H$, and as the rank of the matrix in (\ref{eqjh})
is at least four times the rank of $H$, we have $r^{(4)} \leq r$, and the proof is complete.
For $N = 2^q > 4$, the same process of block column and row operations is repeated $q$ times
resulting in a block upper triangular matrix with the following structure
\begin{equation}
 \left (\begin{array}{ccccc}
H & \mathcal{B}_{1,2} &
 \ldots & \mathcal{B}_{1,N-1} & \mathcal{B}_{1,N}\\
0 & H &
\ldots & \mathcal{B}_{2,N-1} & \mathcal{B}_{2,N}\\
\vdots & \vdots & \cdots & \vdots & \vdots\\
0 & 0 & \ldots & H & \mathcal{B}_{N-1,N}\\
0 & 0 &
 \ldots & 0 & H
\end{array}
  \right)_{Nm\times Nn}\:.
\label{eqrf}
\end{equation}
As the rank of the above matrix, and thus the rank of $\tilde{H}'$,
is at least $N$ times the rank of $H$, we have $r^{(N)} \leq r$.
\end{proof}

\begin{corollary}
If matrix $H$ has full rank, then $r^{(N)} = r$, for $N = 2^q,\:q=1, 2, 3, \ldots$.
\label{corz}
\end{corollary}

\begin{proof}
If $H$ has full rank, then the matrix in (\ref{eqrf}) is also full-rank,
and so is $\tilde{H}'$. This implies $r^{(N)} = r$.
\end{proof}

It is easy to find counter examples to demonstrate that Theorem~\ref{thmz}
and Corollary~\ref{corz} do not always hold for odd values of $N$ or even values
that are not integer powers of two.

\begin{theorem}
If code $\cal{C}$ has minimum distance $d_{min}$, then the minimum distance
$d_{min}^{(N)}$ of a cyclic $N$-lifting $\tilde{\cal{C}}$ of $\cal{C}$ satisfies $d_{min}\leq
d_{min}^{(N)} \leq N\;d_{min}$, for $N = 2^q,\:q=1, 2, 3, \ldots$.
\label{thmx}
\end{theorem}

\begin{proof}
We first prove the lower bound. Consider a codeword $\underbar{c}^{(N)}$ with minimum Hamming weight
$d_{min}^{(N)} = \mathcal{W}_{h}(\underbar{c}^{(N)})$
in $\tilde{\cal{C}}$. Let $\underbar{c}^{(N)} = (\underbar{c}_{0},\underbar{c}_{1},\ldots,\underbar{c}_{N-1})$,
where the subvectors $\underbar{c}_{i},\:i=0, 1,\ldots, N-1$, are of size $n$ each.
Based on $\tilde{H}' \underbar{c}^{(N)} = \underbar{0}$, we have
\begin{equation}
\sum_{i=0}^{N-1} \mathcal{A}_{N-i+j\:\mbox{mod}\:N} \: \underbar{c}_{i} = \underbar{0}\:,
\:\:\: j = 0, \ldots, N-1\:.
\label{eqy}
\end{equation}
Adding all the equations in (\ref{eqy}) for different values of $j$,
and exchanging the order of summations over $i$ and $j$, we obtain
\begin{equation}
\sum_{i=0}^{N-1} \left(\sum_{j=0}^{N-1} \mathcal{A}_{N-i+j\:\mbox{mod}\:N}\right) \underbar{c}_{i}
= \sum_{i=0}^{N-1} \left(\sum_{k=0}^{N-1} \mathcal{A}_{k}\right) \underbar{c}_{i}
= \underbar{0}\:.
\label{eqfd}
\end{equation}
Based on (\ref{eq8}), this implies that $\underbar{c}=\underbar{c}_{0} + \underbar{c}_{1} + \ldots + \underbar{c}_{N-1}$
is a codeword of $\cal{C}$. Moreover, $\mathcal{W}_{h}(\underbar{c}) \leq \mathcal{W}_{h}(\underbar{c}^{(N)}) = d_{min}^{(N)}$.
If $\underbar{c} \neq \underbar{0}\:$, this means $d_{min} \leq d_{min}^{(N)}$.

If $\underbar{c} = \underbar{0}\:$ and $N=2$, then $\underbar{c}_0 = \underbar{c}_1$.
This along with $\tilde{H}' \underbar{c}^{(2)} = \underbar{0}$ results in
$H \underbar{c}_0  = \underbar{0}$, and thus $\underbar{c}_0 \in \cal{C}$.
Since $\underbar{c}_0  \neq \underbar{0}$, this implies $d_{min} \leq  \mathcal{W}_{h}(\underbar{c}_0)
= d_{min}^{(2)} /2 < d_{min}^{(2)}$ and the proof is complete.

If $\underbar{c} = \underbar{0}\:$ and $N > 2$, through a number of steps, we demonstrate that either
there exists a subset $s$ of $Z_{0 \rightarrow N-1}$, such that
the vector $\underbar{v} = \sum_{i \in s} \underbar{c}_{i}$ is nonzero and is in $\cal{C}$,
or $\underbar{c}_{i} = \underbar{c}_{i+N/2}\:,\:\:\:\mbox{for} \:\:
i = 0, 1, \ldots, N/2-1\:.$ For the former case, $d_{min} \leq \mathcal{W}_{h}(\underbar{v})
\leq \mathcal{W}_{h}(\underbar{c}_{s}) \leq
\mathcal{W}_{h}(\underbar{c}^{(N)}) = d_{min}^{(N)}$,
where $\underbar{c}_{s}$ is defined as the
vector of size $|s|\cdot n$ obtained by the concatenation
of vectors $\underbar{c}_{i},\: i \in s$. This proves the lower bound.
For the latter case, the problem is reduced to
that of a lifting with degree $N/2$. Iterating the same process, we either
find a nonzero vector of $\cal{C}$ with Hamming weight less than
$d_{min}^{(N)}$ or reach to a point where all the constituent vectors
$\underbar{c}_{0}, \underbar{c}_{1}, \ldots, \underbar{c}_{N-1}$ of $\underbar{c}$ are equal.
In this case, vector $\underbar{v} = \underbar{c}_{0} =  \underbar{c}_{1} = \cdots =
\underbar{c}_{N-1}$ is nonzero and is in $\cal{C}$.
We thus have $d_{min} \leq \mathcal{W}_{h}(\underbar{v}) = d_{min}^{(N)}/N < d_{min}^{(N)}$.

Here, we explain the first step for demonstrating that when
$\underbar{c} = \underbar{0}\:$ and $N > 2$, either
there exists a subset $s$ of $Z_{0 \rightarrow N-1}$, such that
the vector $\underbar{v} = \sum_{i \in s} \underbar{c}_{i}$ is nonzero and is in $\cal{C}$,
or $\underbar{c}_{i} = \underbar{c}_{i+N/2}\:,\:\:\:\mbox{for} \:\:
i = 0, 1, \ldots, N/2-1\:.$ The other steps are similar and omitted to avoid
redundancy. (Note that the first step suffices to prove the claim for $N = 4$.
For larger values of $N$ further steps are required.)
Consider the vectors
$\underbar{c}'_{i},\:i=0, 1, \ldots, N-1$, defined by
\begin{equation}
\underbar{c}'_{i} = \sum_{j=i}^{i-1+N/2} \underbar{c}_{j\:\mbox{mod}\:N}\:.
\label{eqsg}
\end{equation}
Using equations (\ref{eqy}), it is easy to see that vector
$\underbar{c}'^{(N)} = (\underbar{c}'_{0},\underbar{c}'_{1},\ldots,\underbar{c}'_{N-1})$
satisfies $\tilde{H}' \underbar{c}'^{(N)} = \underbar{0}$, and is therefore in $\tilde{\cal{C}}$.
Moreover, from $\underbar{c} = \underbar{0}\:$, we have
\begin{equation}
\underbar{c}'_{i} = \underbar{c}'_{i+N/2}\:,\:\:\:\mbox{for} \:\:
i = 0, \ldots, N/2-1\:.
\label{eqw}
\end{equation}
Applying this to equation $\tilde{H}' \underbar{c}'^{(N)} = \underbar{0}$,
we obtain
\begin{equation}
\sum_{i=0}^{N/2-1} (\mathcal{A}_{N-i+j\:\mbox{mod}\:N} + \mathcal{A}_{N/2-i+j\:\mbox{mod}\:N}) \: \underbar{c}'_{i} = \underbar{0}\:,
\:\:\: j = 0, \ldots, N/2-1\:.
\label{eqyy}
\end{equation}
Adding the equations in (\ref{eqyy}) for different values of $j$ and switching the summations
with respect to $i$ and $j$, we obtain
\begin{equation}
\left(\sum_{i=0}^{N/2-1} \underbar{c}'_{i}\right) \left(\sum_{k=0}^{N-1} \mathcal{A}_{k}\right)
= \underbar{0}\:,
\label{equr}
\end{equation}
which implies that the vector $\underbar{c}'=\underbar{c}'_{0} + \underbar{c}'_{1} + \cdots + \underbar{c}'_{N/2-1}$
is a codeword of $\cal{C}$. On the other hand, using definition (\ref{eqsg}),
we have $\underbar{c}'= \underbar{c}_{0} + \underbar{c}_{2} + \cdots + \underbar{c}_{N-2}$.
(The subset $s$ in this case is $\{0, 2, \ldots, N-2\}$.)
Thus $\mathcal{W}_{h}(\underbar{c}') \leq \mathcal{W}_{h}((\underbar{c}_{0},\underbar{c}_{2},\ldots,\underbar{c}_{N-2}))
\leq \mathcal{W}_{h}(\underbar{c}^{(N)})$, implying $d_{min} \leq d_{min}^{(N)}$ if $\underbar{c}' \neq \underbar{0}$.
If $\underbar{c}' = \underbar{0}$, then $\sum_{i \in s} \underbar{c}_{i} = \underbar{0}$
over both the even and odd subsets $s$ of $Z_{0 \rightarrow N-1}$. For
$N = 4$, this means $\underbar{c}_0 = \underbar{c}_2$ and $\underbar{c}_1 = \underbar{c}_3$.
Replacing these in (\ref{eqy}), we have
$$
\left
(\begin{array}{c c} \mathcal{A}_0 + \mathcal{A}_2  & \mathcal{A}_1 + \mathcal{A}_3 \\
\mathcal{A}_1 + \mathcal{A}_3 & \mathcal{A}_0 + \mathcal{A}_2
\end{array} \right)
\left
(\begin{array}{c} \underbar{c}_0 \\
\underbar{c}_1
\end{array} \right) = \underbar{0}\:.
$$
Now the problem is reduced to that of $N = 2$, which means
$\underbar{v} = \underbar{c}_0 + \underbar{c}_1 \in \cal{C}$ and either
$\underbar{v} \neq \underbar{0}$, or $\underbar{c}_0 = \underbar{c}_1$.
In both cases, the lower bound is proved.

To prove the upper bound, consider a codeword $\underbar{c} \in \cal{C}$
with $\mathcal{W}_{h}(\underbar{c}) = d_{min}$. Define a vector $\underbar{c}^{(N)}$
of size $n \cdot N$ by $\underbar{c}^{(N)} = (\underbar{c}, \ldots, \underbar{c})$.
It is easy to see based on (\ref{eq8}) that $\underbar{c}^{(N)}$ satisfies
$\tilde{H}' \underbar{c}^{(N)} = \underbar{0}$ and is thus in the
cyclic $N$-lifting $\tilde{\cal{C}}$ of $\cal{C}$. We therefore
have $d_{min}^{(N)} \leq \mathcal{W}_{h}(\underbar{c}^{(N)}) =
N \cdot \mathcal{W}_{h}(\underbar{c}) = N \cdot d_{min}$.
\end{proof}

It is important to note that Theorems 1 and 2 of~\cite{ICV-08} are special cases of
Theorems~\ref{thmz} and \ref{thmx} of this paper, respectively, where $N = 2^q = 2$.

\section{NUMERICAL RESULTS}
\label{sec::simul}
In this section, we apply the IES algorithm of Subsection~\ref{subsec3.2}
to three LDPC codes to eliminate their dominant trapping sets over the BSC.
The codes are: the $(155,64)$ Tanner code~\cite{TSF-01},
a $(504,252)$ randomly constructed regular code~\cite{Mackaycodes},
and an optimized $(200,100)$ randomly constructed irregular code.

\begin{example}
For the $(155,64)$ Tanner code under Gallager B algorithm, the most dominant trapping set is
the $(5,3)$ trapping set, shown in Fig.~\ref{fig::trappingset}, with critical number 3.
We apply the IES algorithm to this code to design cyclic $N$-liftings for
$N = 2, 3, 4,$ and $5$. The FER curves of the designed codes are presented
in Fig.~\ref{fig::tanbsc} along with the FER
of the base code.

A careful inspection of Fig~\ref{fig::tanbsc} shows that using a
2-lifting, the slope of the curve changes from 3 to 4, an indication that
all $(5,3)$ trapping sets are eliminated. In this case,
$(4,4)$ trapping sets play the dominant role. Further increase of $N$ to 3 and then
4, only causes a downward shift of the curve (with no change of slope),
an indication that the minimal critical number remains at 4 for the 2 lifted codes
and increasing the degree of lifting just reduces the number of $(4,4)$
trapping sets. Increasing $N$ to 5 however, eliminates all the
$(4,4)$ trapping sets and the slope of the FER curve
further increases to 5. The dominant trapping sets for the 5-lifting
are $(5,5)$ trapping sets.

It is important to note that for $N = 2$, the performance of the
designed code is practically identical to that of the code designed in Example~3
of~\cite{ICV-08} based on a 2-cover of the Tanner code. There are however no
results reported in~\cite{ICV-08} for covers of larger degree.

For comparison, we have also included in Fig~\ref{fig::tanbsc}, the FER
of a random 5-lifting of the Tanner code. As can be seen, the error floor performance of this
code is significantly worse than that of the designed 5-lifting. In particular, the slope of the random
lifting is just 4 versus 5 for the designed lifting.

The code rates of the designed $N$-liftings are: 0.4065, 0.4043, 0.4032,
and 0.4026, for N = 2 to 5, respectively. The small decrease in the code rate by increasing the
degree of lifting is a consequence of the fact that the original
parity-check matrix of the Tanner code is not full rank. The rate
of the Tanner code itself is $0.4129$.

It is also worth noting that while the girth of the $N$-liftings, $N =2, 3, 4$, remains
the same as that of the Tanner code, i.e., $g = 8$,
for the 5-lifting, the girth is increased to 10.
\label{Tan-BSC}
\end{example}

\begin{figure}[h]
\centering
\includegraphics[width = 0.5\textwidth]{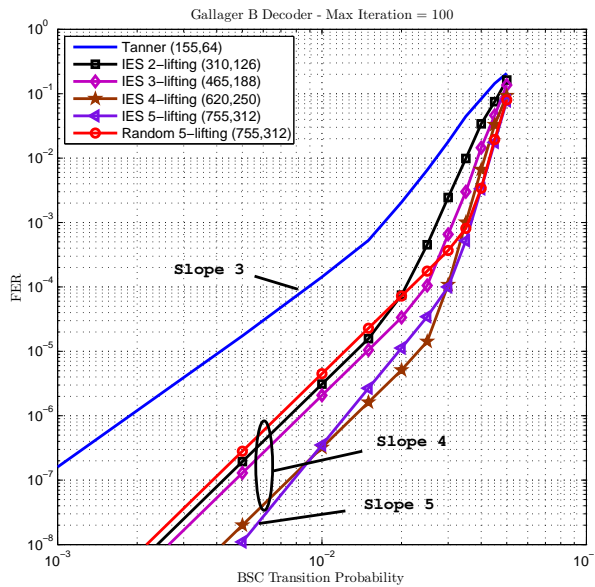}
\caption{Comparison of the FER performance of the Tanner code and
its liftings over the BSC (Example~\ref{Tan-BSC}).}
\label{fig::tanbsc}
\end{figure}

\begin{example}
In this example, we consider a regular $(504,252)$ code
from~\cite{Mackaycodes} decoded by Gallager B algorithm. The dominant trapping sets in this case
have critical number 3 and include $(3,3), (4,2)$, and $(5,3)$ trapping sets among others.
The IES algorithm is used to design cyclic $N$-liftings of this code
for $N = 2$ to 6. The FER results of the liftings and the base code
are reported in Fig.~\ref{fig::macbsc}. Again, the performance of the
$2$-lifting is similar to that of the code designed in~\cite{ICV-08}.
All $(3,3)$ trapping sets are eliminated in the 2-lifting, but the
survival of other trapping sets with critical number 3 keeps the
minimal critical number at 3, and thus no change of FER slope compared to the base code
is attained. Increasing $N$ to 3, however, eliminates all the trapping sets with
critical number 3 and changes the slope of the FER to 4. The dominant trapping sets in this
case are $(4,4)$ sets. Further increase of $N$ to 4 and 5 only reduces the number
of $(4,4)$ trapping sets and thus results in a downward shift of the FER curve.
For $N=6$, the algorithm can eliminate all the $(4,4)$ trapping sets, and thus
increases the slope of the FER curve to 5. The dominant trapping sets
in this case are $(5,5)$ sets.

For comparison, in Fig.~\ref{fig::macbsc}, we have also shown the performance of
a random 6-lifting of the $(504,252)$ code. As can be seen the performance of this
code in the error floor region is far poorer than that of the designed $6$-lifting.
In particular, the slope of the FER curve for this code is only 3 versus 5
for the designed code.

In this example, the parity-check matrix of the base code is full-rank,
and all the liftings have the same rate of 0.5 as the base code.

Noteworthy is that while the 2-lifting has the same girth of $g=6$ as the base code,
the girth for $N$-liftings, $N$ = 3 to 6, is increased respectively to 8, 8, 8 and 10.
\label{Mc-BSC}
\end{example}

\begin{figure}[t]
\centering
\includegraphics[width = 0.5\textwidth]{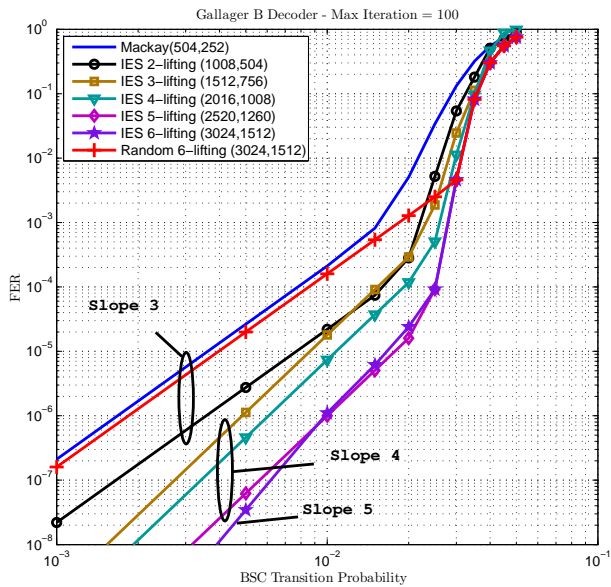}
\caption{Comparison of the FER performance of the regular
$(504,252)$ code and its liftings over the BSC
(Example~\ref{Mc-BSC}).} \label{fig::macbsc}
\end{figure}

\begin{example}
In this example, we consider a randomly constructed rate-1/2 irregular $(200,100)$ code
as the base code.
The degree distributions for this code, optimized for Gallager A algorithm
over the BSC~\cite{BRU-04}, are $\lambda(x) = 0.1115 x^2 + 0.8885 x^3$
and $\rho(x) = 0.26 x^6 + 0.74 x^7$. The code has $g = 6$. This code has a wide variety
of dominant trapping sets under Gallager A algorithm, all with critical number 3.

We apply the
IES algorithm to this code to design a cyclic $13$-lifting
of length $2600$, rate 0.5 and $g=6$. The FER curves of the lifted code and the base code
are presented in Fig.~\ref{fig::irrBSC}. As can be seen, the lifted code has a much better
error floor performance compared to the base code. In fact, the minimum critical number
for the lifted code is 5 versus 3 for the base code. For comparison, a rate-1/2 code of block length $2600$
with the same degree distribution and $g = 6$ is constructed. The performance of this code
is also given in Fig.~\ref{fig::irrBSC}. Clearly the performance of the lifted code is
significantly better in the error floor region. In particular, the slope of the FER curve
for the random code is the same as the base code and much less than that of the lifted code.
\label{irr-BSC}
\end{example}

\begin{figure}[h]
\centering
\includegraphics[width = 0.5\textwidth]{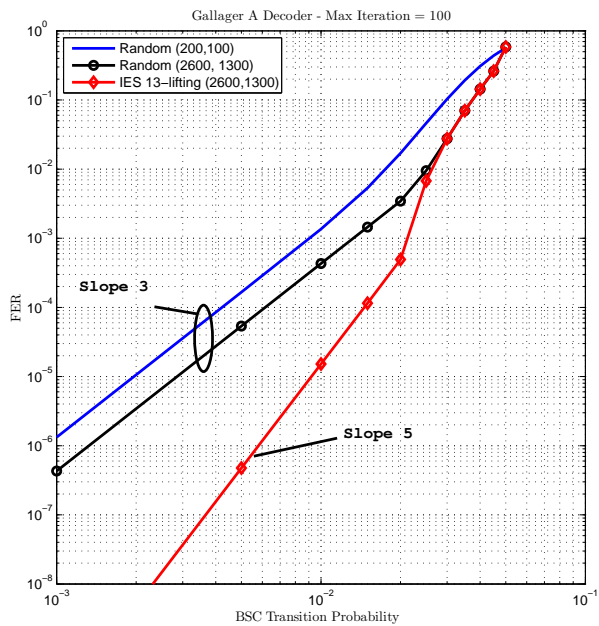}
\caption{FER curves of the irregular $(200,100)$ code, its IES
cyclic 13-lifting and a random irregular $(2600,1300)$ code with the
same degree distribution over the BSC (Example~\ref{irr-BSC}).}
\label{fig::irrBSC}
\end{figure}

\begin{example}
It is known that codes designed for a certain channel/decoding algorithm
would also perform well for other channel/decoding algorithms~\cite{FFR-06}. In this example,
we show that cyclically lifted codes designed for Gallager B algorithm in Examples
\ref{Tan-BSC} and \ref{Mc-BSC} also perform very well on the binary-input AWGN
channel under min-sum algorithm. The FER results for the 5-lifting of the Tanner code and
the 6-lifting of the MacKay code are reported in Figures \ref{fig::tanawgn} and \ref{fig::macawgn}, respectively.
In each figure, the performance of the corresponding base code
and a similar random code (same block length and degree distributions)
is also presented. One can see that at high SNR values, the designed codes perform
far superior to the corresponding base codes and random codes. In particular, they show no sign of
error floor for FER values down to about $10^{-8}$, and their FER decreases
at a much faster rate compared to the base codes and random codes.
\label{AWGN}
\end{example}

\begin{figure}[h]
\centering
\includegraphics[width = 0.5\textwidth]{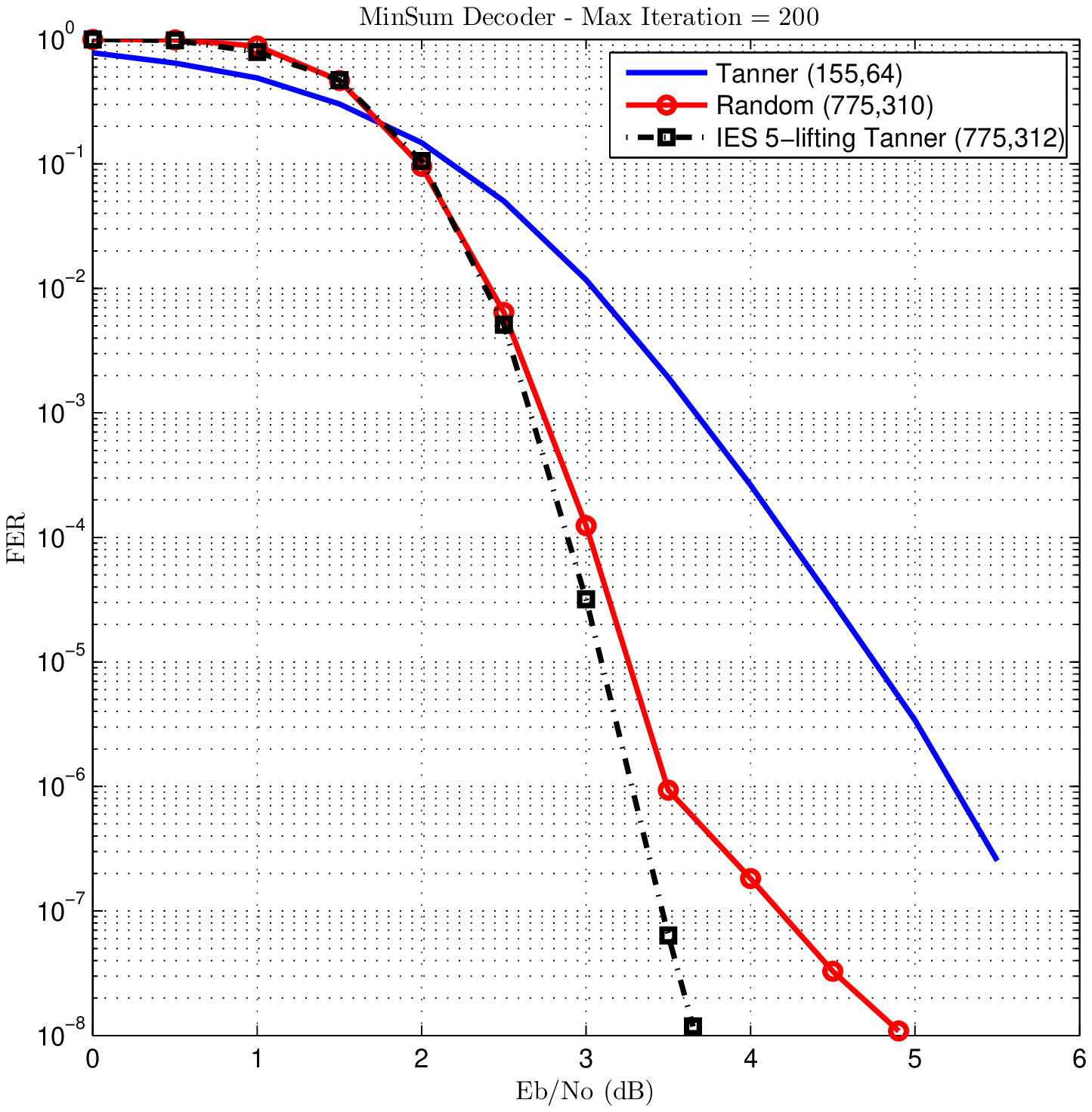}
\caption{FER performance of the Tanner code, its IES cyclic
5-lifting and a random $(775,310)$ code with the same degree
distribution over the BIAWGN channel (Example~\ref{AWGN}).}
\label{fig::tanawgn}
\end{figure}

\begin{figure}[h]
\centering
\includegraphics[width = 0.5\textwidth]{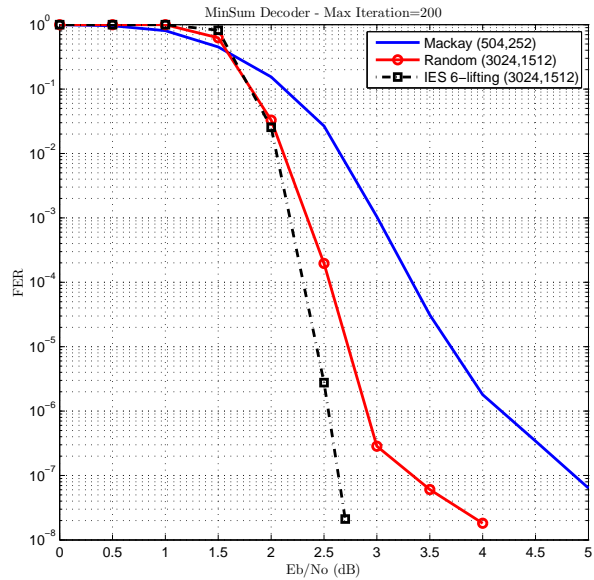}
\caption{FER performance of the regular $(504,252)$ code, its IES
cyclic 6-lifting and a random $(3024,1512)$ code with the same
degree distribution over the BIAWGN channel (Example~\ref{AWGN}).}
\label{fig::macawgn}
\end{figure}

\section{CONCLUSIONS}
\label{sec::conclusion} In this work, cyclic liftings are proposed to
improve the error floor performance of LDPC codes.
The liftings are designed to eliminate the dominant trapping sets of
the code by eliminating their constituent short cycles. The design
approach is universal in that it can be applied to any decoding algorithm
over any channel, as long as the dominant trapping sets are known and available.
In addition, the liftings have the same degree distribution as the
base code and are implementation friendly due to their cyclic
structure. For base codes with full-rank parity-check matrices,
the liftings also have the same rate as the base code and the performance improvement
is achieved at the expense of larger block length. Compared to
random codes or random liftings with the same block length and degree
distribution, the designed codes perform significantly better in the error floor region.

While the cyclic liftings in this work were designed for
Gallager A/B algorithms over the BSC, they also performed very
well over the BIAWGN channel. In particular, the designed codes substantially outperformed
similar random codes in the high SNR region.

\section*{Acknowledgment}
The first and the third authors would like to thank Dr. Hassan Haghighi from Mathematics
Department of K. N. Toosi University of Technology for helpful discussions
on the material presented in Subsection~\ref{subsec3.1} of the paper.

\end{document}